\documentclass[sigconf]{acmart}
\AtBeginDocument{%
  }

\setcopyright{acmlicensed}
\copyrightyear{2025}
\acmConference[Conference acronym 'XX]{Make sure to enter the correct
  conference title from your rights confirmation email}{June 03--05,
  2018}{Woodstock, NY}
\acmISBN{978-1-4503-XXXX-X/18/06}

\usepackage{algorithm}
\usepackage{algorithmic}
\usepackage{subcaption}
\usepackage[shortlabels]{enumitem}
\usepackage{shortcuts}

\usepackage{cleveref}
\crefname{theorem}{Thm.}{Thms.}
\crefname{lemma}{Lem.}{Lemmas}
\crefname{corollary}{Cor.}{Cors.}
\crefname{figure}{Fig.}{Figs.}
\crefname{definition}{Defn.}{Defns.}
\crefname{table}{Tab.}{Tabs.}
\crefformat{section}{\S#2#1#3}
\crefmultiformat{section}{\S#2#1#3}{ and~\S#2#1#3}{, \S#2#1#3}{ and~\S#2#1#3}
\crefname{example}{Ex.}{Exs.}
\crefname{item}{item}{items}
\crefname{footnote}{footnote}{footnotes}
\crefname{observation}{Obs.}{Obs.}
\crefname{remark}{Remark}{Remarks}
\crefname{proposition}{Prop.}{Props.}
\crefname{equation}{Eqn.}{Eqns.}
\crefname{counterexample}{Counterexample}{Counterexamples}
\crefname{property}{Property}{Properties}
\crefname{algorithm}{Algorithm}{Algorithms}

\newcommand{\pfuzz}{\textsc{PathFuzzing}}

\setlist{nosep,leftmargin=\parindent}
\begin{document}

\title{\pfuzz{}: Worst Case Analysis by Fuzzing Symbolic-Execution Paths}


\author{Zimu Chen}
\affiliation{%
  \institution{Peking University}
  \city{Beijing}
  \country{China}}
  
\author{Di Wang}
\affiliation{%
  \institution{Peking University}
  \city{Beijing}
  \country{China}}

\renewcommand{\shortauthors}{Zimu et al.}

\begin{abstract}
Estimating worst-case resource consumption is a critical task in software development.
The worst-case analysis (WCA) problem is an optimization-based abstraction of this task.
Fuzzing and symbolic execution are widely used techniques for addressing the WCA problem.
However, improving code coverage in fuzzing or managing path explosion in symbolic execution within the context of WCA poses significant challenges. 

In this paper, we propose \pfuzz{}, aiming to combine the strengths of both techniques to design a WCA method.
The key idea is to transform a program into a symbolic one that takes an execution path (encoded as a binary string) and interprets the bits as branch decisions.
\pfuzz{} then applies evolutionary fuzzing techniques to the transformed program to search for binary strings that represent satisfiable path conditions and lead to high resource consumption.
We evaluate the performance of \pfuzz{} experimentally on a benchmark suite that consists of prior work's benchmarks and some added by us.
Results show that \pfuzz{} generally outperforms a fuzzing and a symbolic-execution baseline.
\end{abstract}



%\keywords{Worst Case Analysis, Fuzzing, Fuzz Testing, Symbolic Execution}

%% \received{20 February 2007}
%% \received[revised]{12 March 2009}
%% \received[accepted]{5 June 2009}

\maketitle

\section{Introduction}
\label{sec:intro}

\begin{comment}
    
In software engineering, understanding the resource consumption of programs---especially their worst-case cost---is critical throughout software development.
%
One of the most common forms of resource consumption is the time spent on program execution, while other forms may include memory usage, network bandwidth occupation, or energy consumption.

Although most software engineers are careful about resource consumption, they can still be careless sometimes, especially in large-scale software projects.
%
Ignorance of important resource bottlenecks may lead to unexpected excessive resource consumption.
%
Given these considerations, it is important and useful to automate the analysis of the resource consumption of programs.

In this work, we study an abstraction of resource consumption analysis: the Worst Case Analysis (WCA) problem, defined as follows:

\begin{definition}[Worst Case Analysis]
    Given a program $P$, its input space $I_P$ and the cost metric $\mathrm{cost(\cdot)}:I_P\rightarrow \mathbb R ^+$, Find a specific input $\alpha$ that maximizes the resource consumption $\mathrm{cost}(\alpha)$. Formally, find:
    $$
    \arg \max_{\alpha \in I_P} \mathrm{cost}(\alpha)
    $$
\end{definition}
\end{comment}

The \emph{worst-case-anlaysis} (WCA) problem is a widely-studied problem abstraction of the resource-consumption analysis. Formally, given a program $P$, an input space $I_P$, and a cost metric $\mathrm{cost(\cdot)}:I_P\rightarrow \mathbb R ^+$, the goal of WCA is to find a specific input $\alpha$ that maximizes (often approximately) the resource consumption $\mathrm{cost}(\alpha)$.
Multiple methods have been used in designing WCA problem solvers. Two common methodologies are fuzzing and symbolic execution:
%
%They have different advantages and drawbacks and are used to address different types of programs:
%
\begin{itemize} 
\item
\emph{Fuzz Testing} is a widely adopted methodology in software testing.
When applying fuzz testing to solve WCA, the main idea is to iteratively generate new inputs from existing ones, evaluate the program's resource consumption under these inputs, and keep candidates that may lead to high resource consumption according to some heuristics.
In this category, evolutionary fuzzing is one commonly used method for solving WCA.

\item
\emph{Symbolic Execution} is an analysis methodology that uses symbolic values as inputs to execute the program. 
When applying symbolic execution to solve WCA, the main idea is to enumerate the program's different execution paths with some search strategies.
When the number of different paths is much less than the number of possible inputs, symbolic execution usually provides an efficient solution to WCA.
\end{itemize}

In general, fuzz testing is good at recognizing patterns of target inputs but may fail to locate corner-case inputs with high resource consumption.
On the other hand, symbolic execution is good at finding corner cases but may be ineffective due to the path explosion issue.
Recently, hybrid methods that integrate fuzz testing and symbolic execution have been used for WCA problems, aiming to combine the advantages of both methods. 
\emph{Badger}~\cite{ISSTA:NKP18} has provided a hybrid WCA method that runs a symbolic-execution engine and a fuzzing engine in parallel on the same program while continuously exchanging found inputs between the two engines.
In some sense, \emph{Badger} serves as a framework to make different optimizers cooperate; it does not aim to advance either the fuzzing or the symbolic execution part.
%
%The experiments have been done on severgal programs in Java, using KelinciWCA (a fuzzing tool for Java based on AFL) and Symbolic Pathfinder (a symbolic execution tool for Java).
%%
%In most cases, Badger outperformed the individual use of either of the components.
%
\emph{ESE}~\cite{ISSRE:ADS18} has proposed a new evolutionary-fuzzing scheme that treats the path-condition sequences of a program as its search space.
During each evolutionary iteration, \emph{ESE} runs symbolic execution with respect to a candidate path-condition sequence; if some branch condition cannot be resolved, \emph{ESE} selects a branch randomly.
However, performing operations (such as mutation and crossover) on the path-condition sequences may break the dependencies between symbolic values, e.g., a symbolic value in a conditional statement takes some value only if another symbolic value takes some value that leads to a corresponding branch.
%
%In this work, a path is represented as a sequence of path conditions. In the crossover operation, offspring are generated by combining random halves of the parents.
%
%Symbolic execution is then applied to select a random path that satisfies the new path condition sequence. If there is none, the offspring is discarded.

%Other methods have also combined fuzzing and symbolic execution to create more efficient methods. These methods have explored various framework to merge the component methods and take advantage analyzing specific type of programs, usually better than either fuzzing or symbolic execution themselves.

%However, present hybrid methods lacks internal integration between their component methods. Take Badger as an example, the way it combines fuzzing and symbolic execution determines that for each component, the other is a black-box where no information but the transferred input can be utilized. Therefore, substantial opportunities exist for developing more sophisticated integration frameworks for WCA solvers.

%
In this paper, we propose \pfuzz{}, a new evolutionary-fuzzing scheme that integrates with symbolic execution to solve WCA.
Our basic approach follows that of \emph{ESE}: perform fuzzing in the search space of program paths instead of the original input space.
Our key idea is to \emph{transform the program into a symbolic one that takes an execution path as its input}.
In this way, applying evolutionary fuzzing to the transformed program immediately yields a WCA solver that integrates fuzzing with symbolic execution.
Different from \emph{ESE}, we do not use path-condition sequences to encode execution paths; instead, we use \emph{binary strings} to encode them, each bit of which represents a branch selection.
Our encoding has two benefits: (i) it avoids the issue of symbolic-value dependencies, and the symbolic-execution engine can easily reconstruct a path-condition sequence from the binary string; and (ii) it simplifies the representation of the search space, and provides the potential for the evolutionary algorithms to identify patterns in terms of branch selections.
Based on this insight, we develop some optimization heuristics to improve the performance of \pfuzz{}.
When the symbolic execution follows a binary string to execute a program, if it finds out that a branch is determined (e.g., by checking if current path conditions yield a tautology), it does not need to consume a bit from the binary string.
Intuitively, this optimization allows \pfuzz{} to exploit each candidate binary string as much as possible.
We also consider the problem of reducing the number of satisfiability checks, which are the common performance bottleneck when combining fuzzing and symbolic execution.
We present a preliminary finding that for some programs in which branch selections are independent of symbolic inputs, \pfuzz{} can largely reduce the number of satisfiability checks.

%\begin{itemize}
%    \item
%    Transform the original program $P$ into a new program $P'$.
%    The input of program $P$ is a binary string $q \in \{0,1\}^{M}$, and its resource consumption is defined by the time spent on symbolically executing $P$ through the path defined by string $q$ (A formal definition will be given in \cref{sec:method}).
%    After this transformation, the original problem is transformed into WCA on another program $P'$, which takes a binary string as input.
%    \item
%    We apply a fuzzing-based method on program $P'$.
%    A binary string $q_0$ will be found, which represents an execution path on program $P$.
%    \item 
%    For the path $q_0$, we use an SMT solver to find a concrete input $\alpha_0$. 
%    The concrete input found is an lower-approximation of the worst-case input of the original problem $P$.
%\end{

The main contribution of our work is as follows:

\begin{itemize}
    \item 
    We present an encoding of program-execution paths as binary strings and a transformation from a program to its symbolic variant that takes such binary strings as inputs.
    \item
    We propose \pfuzz{}, a hybrid method addressing the WCA problem that combines evolutionary fuzzing and symbolic execution based on the binary-string encoding of the search space.
    \item We have conducted experiments to evaluate the method. Results show that \pfuzz{} generally outperforms a fuzzing baseline and a symbolic-execution baseline and has an advantage on certain programs where the optimization heuristics work.
    
\end{itemize}

%\paragraph{Chapter Overview} The contents of this paper is organized as follows. In Chapter 2, we use an example to show the problems fuzzing and symbolic-execution based methods are faced with dealing with WCA problems, and how our method could potentially work better. In Chapter 3, we present the complete workflow of our method, including the definition of cost metric and path strings, the conversion of the problem, and the design of evolutionary algorithm. In Chapter 4, we describe our implementation of the method and the setup of the experiment. In Chapter 5, we show the results of our experiments. In Chapter 6, we present some discussion based on the experiment result, along with threats to validity. In Chapter 7, we show the related work about solving the WCA problem using fuzzing, symbolic-execution or compound methods. In Chapter 8, we conclude this work and give directions for future work.

\section{Motivation}
\label{sec:motivation}

In this section, we provide a motivating example by the QuickSort program below:
 
\begin{algorithm}[!h]
    \caption{QuickSort (Simple Impl.) (N)}
    \label{alg:AOA}
    \renewcommand{\algorithmicrequire}{\textbf{Input:}}
    \renewcommand{\algorithmicensure}{\textbf{Output:}}
    \begin{algorithmic}[1]
        \REQUIRE $N, A: \texttt{Array}[\texttt{int}[1, N]] \ \mathrm{where} |A| = N$

        \IF{$N \le 1$}
            \RETURN $A$
        \ENDIF

        \STATE AddCost($N$)

        \STATE $pivot \leftarrow A[0]$
        \STATE $left, mid, right \leftarrow [], [A[0]], []$

        \FOR{$i \in 1 \dots N-1$}
            \IF{$A[i] < pivot$}
                \STATE $left \leftarrow left + [A[i]]$
            \ELSIF{$A[i] = pivot$}
                \STATE $mid \leftarrow mid + [A[i]]$
            \ELSE
                \STATE $right \leftarrow right + [A[i]]$
            \ENDIF
        \ENDFOR

        \RETURN Quicksort($|left|, left$) $+ mid +$ Quicksort($|right|, right$)
        
    \end{algorithmic}
\end{algorithm}

In the program above, AddCost($N$) indicates the cost metric of the program: a call to function Quicksort with an input array of size $N$ costs $N$ unit amount of time for $N \ge 2$ (does not include the recursive call part). This value has a linear relationship to the actual running time and may be considered an abstraction of it.

Given the program, we want to know for a fixed $N$, what input array $A$ maximizes the cost metric of the program. Thus, the setup of symbolic execution is $N=N_0, A=[a_0, \dots, a_{n-1}]$, where $N_0$ is a fixed concrete value like $64$ that does not change during the entire optimization process, and $a_0, \dots, a_{n-1}$ are symbolic values that we want to optimize.

We first show some algorithmic properties about this time-complexity vulnerable QuickSort program. The worst case complexity $T(n) = T(n-1) + O(n) = O(n^2)$ can be triggered by having all non-pivot elements put in one side ($left$ or $right$) in every recursive call. Furthermore, without pivot randomization, we know that the first element is the pivot, so we can simply let all other elements be less than the first element or let all other elements be greater than the first element in a recursive call.

The worst case input easiest to find is the decreasing or increasing array: $[1, 2, \dots, N]$ or $[N, N-1,\dots, 1]$. There are other worst case inputs like $[1, N, N-1, \dots, 2]$, or sufficiently, every element in the array must be greater (or less) than all succeeding elements.

\paragraph{Path Strings} To show how \pfuzz{} solves this problem efficiently, let's first consider representing a program execution path with a binary string.

Assume we have an extra dynamic array $path$ in program state.
For every branch in the program (the two \textbf{if} branches), if the branch is taken (path condition evaluates to $\texttt{True}$), then append a $\texttt{1}$ to $path$, otherwise, append a $\texttt{0}$ to $path$. 

Notice that the \textbf{for} loop, which should be replaced by an equivalent \textbf{while} loop to avoid using advanced language features like iterators, also has a branch with path condition $i \le N-1$. 

This branch, however, is not considered (do not add a bit after evaluating loop condition) here since $N$ is a concrete value; the for loop could actually be unrolled.

For example, input $N=8,A=[3,1,4,5,3,2,2,3]$ has the following path string in the first function call (recursive calls not included) to QuickSort (we use space to separate branches encountered in each loop):

$$
\texttt{1 00 00 01 1 1 01}
$$

In each iteration, path strings $\texttt{1}, \texttt{01}, \texttt{00}$ represent that the current value is less than, equal to, or greater than the pivot, respectively.

\paragraph{Program Transformation} We also notice that inputs with the same path string have the same control flow, thus the same cost, as at statement AddCost($N$), the value of $N$ is always a concrete value. The cost of a path string can be obtained by a symbolic execution with this path string, in which each bit of the path string guides the symbolic executor at at branch, whether it should take the branch or not.

Some path strings, however, do not represent the control flow of any program input since the path conditions on that path are unsatisfiable. This can be checked using an SMT solver. In this work, we use Z3\cite{TACAS:Z308}, a high-performance SMT solver developed by Microsoft Research that supports automated reasoning with arithmetic, bit vectors, and other theories. Unsatisfiable path strings should be discarded since they do not help in finding worst-case input.

At this point, the symbolic executor can be considered a new program $P'$ with cost metric, whose input set is the set of binary strings $\{\texttt{0},\texttt{1}\}^M$. To find the worst case input, we can instead find the worst case path string, which is equivalent to the WCA problem on $P'$.

\paragraph{Evolutionary Algorithm} Then, \pfuzz{} utilizes fuzzing technique on program $P'$, using an evolutionary algorithm specifically designed for programs with binary string as input to find the worst case input for $P'$.

According to our previous observation, a worst case path (in a function call) should be either $\texttt{(1)}^{N-1}$ (repeating $\texttt{1}$ for $N-1$ times, same for below), i.e. all other elements are less than pivot, or $\texttt{(00)}^{N-1}$, i.e. all other elements are greater than pivot. For the whole program execution (including recursive calls), all possible paths of the worst-case input are as follows:

$$
s_1^{N-1} s_2^{N-2} \dots s_{N-1}^1  \ \ \  \forall i =1\dots N-1, s_i \in \{\texttt{1}, \texttt{00}\}
$$

For the two simplest forms of worst case input, i.e. the ascending sequence and the descending sequence, their path string is $\texttt{(00)}^{N(N-1)/2}$ and $\texttt{(1)}^{N(N-1)/2}$ respectively. 

An evolutionary algorithm could easily find such path strings consisting of the same value on each bit. Roughly speaking, the algorithm may find that the more $\texttt{1}$ there is, the higher the cost of the path. The crossover operator in the evolutionary algorithm may combine two paths where $\texttt{1}$ appears at high frequency and generates another with an even higher, eventually resulting in the all $\texttt{1}$ string. Notice that the algorithm does not need to find the string with exactly length $N(N-1)/2$: any string with prefix $\texttt{(1)}^{N(N-1)/2}$ is OK.

\paragraph{Input Model Solving} Finally, with worst case path and path conditions on that path, we can use SMT solver to find a model for the set of path conditions. As guaranteed before, it is satisfiable and we will find an input, which is the worst case input we're looking for.

\section{Method}
\label{sec:method}

In this section, we introduce the general framework of our method, \pfuzz{}.
We use the following symbols for the program we aim to analyze and its related properties:
\begin{itemize}
    \item We use symbols like $P, P'$ to represent a program. 
    
    \item Set $I_P$ represents the input space of program $P$. $\alpha \in I_P$ represents an exact input of the program.

    \item Program $P$ consists of a set of statements $\{s_1, s_2, \dots, s_n\}$. For simplicity, let $s_1$ always be the entry point for a program. 
    
    \item Each statement $s_i$ has some successor statements, represented by a set of statements $\mathrm{OUT}(s_i)$, where $|\mathrm{OUT}(s_i)| \in \{0,1,2\}$

    \item The cost metric is defined by a function $\mathrm{cost}_P(stmt, state)$, representing the amount of resource consumption that should be added after executing statement $stmt$ with program state $state$.
    
    \item The function $\mathrm{eval}_P(\alpha)$ represents the cost (resource consumption) of running program $P$ with input $\alpha$.
\end{itemize}

These considerations give a simple ``basic block'' model to the programs.
We assume that the program we want to analyze does not contain I/O operations, multi-threads, exceptions, or other advanced language features.
We will introduce some strategies in the experiment sections (\cref{sec:evaluation-setup,sec:evaluation-result}) to make our method compatible with complex real-world programs.

\subsection{Program Path}
\label{sec:program-path}

We first give a more formal definition to \emph{Execution Path}, which is a model of actual program executions. Then, we introduce \emph{Path String}, a data structure implicitly representing an execution path and directly used in fuzzing.

\subsubsection{Execution Path}

\begin{definition}[Execution Path]
     An execution path is a sequence of statements $T = \{t_1, t_2, \dots, t_l\}$ that:
\begin{itemize}
    \item $t_1$ is the entry statement of the program,
    \item $\forall 1 < i \le l, t_i \in \mathrm{OUT}(t_{i-1})$, and
    \item $|\mathrm{OUT}(t_l)| = 0$. 
\end{itemize}
\end{definition}

The second condition ensures that the path follows some valid control flow sequence of the program $P$.
The third condition guarantees that the sequence reflects a terminated (not on-the-way) program execution.

For a program $P$, if we run it with input $\alpha \in I_P$ and the program executes through the execution path $T$, then we say that:

\begin{itemize}
    \item The input $\alpha$ \textbf{is related to} the execution path $T$.
    \item The execution path $T$ is \textbf{satisfiable}. We use the notation $\mathrm{sat}(T) = \texttt{True}$ to represent it.
\end{itemize}

\noindent Obviously, an input $\alpha$ is exactly related to one execution path; we use notation $T(\alpha)$ to represent that path.
However, an execution path $T$ is related to either none, one, or more inputs.
We use $I_P(T)$ to represent the set of that input.

For execution path $T$, if no input $\alpha$ is related to it, then we say it's \textbf{unsatisfiable} and use the notation $\mathrm{sat}(T) = \texttt{False}$ to represent it. This is the case when the path conditions on that path cannot be satisfied. For example:

\begin{algorithm}[!h]
    \caption{Unsatisfiability Example}
    %% \label{alg:AOA}
    \renewcommand{\algorithmicrequire}{\textbf{Input:}}
    \renewcommand{\algorithmicensure}{\textbf{Output:}}
    \begin{algorithmic}[3]
        \REQUIRE $X: \texttt{int}$

        \IF {$X > 0$}
            \IF {$X < 0$}
                \RETURN
            \ENDIF
        \ENDIF

    \end{algorithmic}
\end{algorithm}

The program has three distinct execution paths.
However, the path that takes the $ \texttt{True}$ branch for both $\mathbf{if}$ statements is unsatisfiable, since no $X$ satisfies $X > 0 \wedge X < 0$.

Here, we define some other properties of execution paths:

\begin{definition}[Cost of Input and Execution Path]
     For Program $P$, input set $I_P$, and cost metric $\mathrm{cost}_P(stmt,state)$, the cost of an input $\alpha \in I_P(T)$ is
     $$
     \mathrm{cost}(\alpha) = \sum_{i=1}^l \mathrm{cost}(t_i, state(\alpha, i)) ,
     $$
     where $state(\alpha, i)$ is the program state where the program $P$ is executed with input $\alpha$, and has executed through statements $t_1, \dots ,t_{i-1}$.
     The cost of execution path $T$ is the maximum cost among all possible inputs:
     $$
     \mathrm{cost}(T) = \max_{\alpha \in I_P(T)} \mathrm{cost}(\alpha)
     $$
\end{definition}

In the above definition, if $\mathrm{cost}$ is constant, then $\mathrm{cost}(\alpha)$ has the same value for all $\alpha \in I_P(T)$ since $\mathrm{cost}$ has nothing to do with $state(\alpha, i)$.
Thus, the definition can be simplified to: 
$$
\mathrm{cost}(T) = \sum_{i=1}^l \mathrm{cost}(t_i)
$$
This indicates that given a specific execution path, we can use symbolic execution to find the cost of that path without knowing a specific input.
It becomes a little bit more complex when $\mathrm{cost}$ is variable. ]
In this case, we can also use symbolic execution to find some $\mathrm{cost}'(T)$, however it will be a function (or practically, a symbolic expression) of input variables (in the form like: $n^2 + 2m + 10$ where $n,m$ are input variables).
We can use an SMT solver to find the maximum value among all inputs $\alpha \in I_P(T)$, then find the concrete value of $\mathrm{cost}(T)$.

\subsubsection{Path String}

On an execution path $T$, for all $i$, if $|\mathrm{OUT}(t_i)| = 1$, then we immediately know $t_{i+1}$. The only thing we care about is which side is taken on each branch operation. According to this intuition, we define the path string, a simple data structure that indirectly reflects an execution path.

\begin{definition}[Path String]
     A path string is a binary string $q = q_1q_2 \dots q_m, q_i \in \{0,1\}$.
\end{definition}

A path string itself is just a data structure and has nothing to do with the program. We need a mapping function $T(q) = t_1, t_2, \ldots, t_l$ to define which actual path is represented by this path string.
Intuitively, we have a brief idea: $q_1, q_2, \ldots$ means whether the 1st, 2nd, \ldots branches are taken. Formally, we define the mapping function that defines which execution path is represented by a path string.

\begin{definition}[Default Path String Mapping]
     For program $P$, the default path string mapping $T_d(q) = t_1, t_2, \ldots, t_l$ is defined as follows:
     
    \begin{itemize}
        \item $t_1 = c_1$.
        \item Iteratively, $\forall i \ge 2:$
        \begin{itemize}
            \item If $|\mathrm{OUT}(t_{i-1})| = 0$, then $t_i = \mathrm{NULL}$. (No such statement, the program terminated after statement $i-1$.)
            \item If $|\mathrm{OUT}(t_{i-1})| = 1$, then $t_i = \mathrm{OUT}(t_{i-1})_1$.
            \item If $|\mathrm{OUT}(t_{i-1})| = 2$, then $$
            t_i = \begin{cases}
            \mathrm{OUT}(t_{i-1})_1 & \text{if } q_{x} = 0, \\
            \mathrm{OUT}(t_{i-1})_2 & \text{if } q_{x} = 1,
            \end{cases}
            $$
            where $x = 1 + \sum_{j=1}^{i-1}[|\mathrm{OUT}(t_j)| = 2]$.
        \end{itemize}
    \end{itemize}
    
\end{definition}

\noindent The ``where'' clause in the above definition requires $x \le |q|$, i.e., the path string should never be not long enough so that there are no more bits to use while the program is not yet terminated.
Intuitively, if $|q|$ is sufficiently large, then there would be no problem. We ignore this requirement here, and a later condition will ensure that this condition is always true.

It should be noted that we will provide other mapping functions in the optimization section.
In the following part, we consider that we are using a certain mapping function $T(q)$.

\begin{definition}[Properties of a Path String]
    Given program $P$, input set $I_P$, cost metric $\mathrm{cost}_P(stmt,state)$ and mapping function $T(q)$, for some path string $q$:
    \begin{itemize}
        \item $q$ is satisfiable if $T(q)$ is satisfiable. Otherwise, we say that $q$ is unsatisfiable.
        \item If $q$ is satisfiable, then the cost of $q$ is defined as $\mathrm{cost}(q) = \mathrm{cost}(T(q))$.
        \item Two strings $q_1, q_2$ are considered equivalent if and only if $T(q_1) = T(q_2)$. This defines an equivalence relation on path strings.
        \item If a input $\alpha$ is related to $T(q)$, then we also say it is related to $q$ (actually, the equivalence class where $q$ is in). Let $I_P(q) = I_P(T(q))$, and $Q(\alpha) = \{q:\alpha \in T(q)\}$.
    \end{itemize}
\end{definition}

At the end of this section, for convenience, we show that we can fix the length of path strings to some certain threshold $M$ without loss of expressiveness of path strings.

\begin{theorem}[Path String Length Threshold]
     For program $P$, input set $I_P$ and mapping function $T(q)$, there exists a positive integer $M$, so that 
     $$
     \bigcup_{q \in \{0,1\}^M} I_P(q) = I_P
     $$
\end{theorem}
\begin{proof}
    Let $M = \max_{\alpha\in I_P} \min_{q \in Q(\alpha)} |q|$. For all $\alpha \in I_P$, there exists a path string $q \in Q(\alpha)$, where $|q| \le M$. It's obvious that by appending certain number of $\texttt{0}$ after $q$, we get an equivalent path string $q'$, and we have $q' \in \{0,1\}^M$ and $\alpha \in I_P(q')$. This shows that such $M$ ensures the condition required in the theorem.
\end{proof}

With $M$ defined, it's simple to prove that the additional requirement in the definition of path mapping is always satisfied with condition $|q| = M$ appended.

\subsection{Problem Conversion}
\label{sec:problem-conv}

With the definition of path strings and related properties and operations, we can finally propose the idea of how to convert the problem.
First, we give a definition to the new problem:

\begin{definition}[Worst Case Path Analysis]
    Given a program $P$, its input space $I_P$, a path string length threshold $M$ and the cost metric $\mathrm{cost(\cdot)}:I_P\rightarrow \mathbb R ^+$, Find a specific path string $q$ that is satisfiable and maximizes the resource consumption of any input related to that path. Formally $\mathrm{cost}(q)$. Formally, find:
    $$
    \arg \max_{q \in \{0,1\}^{M}, \mathrm{sat}(q)} \mathrm{cost}(q)
    $$
\end{definition}

In the part where we have defined the cost of an execution path, we showed that the evaluation of $\mathrm{cost}(q)$ could be done by symbolically executing the program.
The value of $\mathrm{sat}(q)$ could be found by using an SMT solver to check the satisfiability of the conjunction of all path conditions on that path.
Thus, we can use certain search strategies to search the space $\{0,1\}^{M}$ and find the (approximately) maximum cost and the path $q$, while in this work, the evolutionary algorithm will become the main search method.
After $q$ is found, we can use an SMT solver to find a concrete input $\alpha$ that satisfies all path conditions on $T(q)$, which means $\alpha \in I_P(q)$ and is the worst case input we are looking for; however, if $\mathrm{cost}$ is variable, we instead need to find the $\alpha$ that maximizes $\mathrm{cost}(\alpha)$, which is harder but also solvable by an SMT solve.

On an alternate view, to stress the use of fuzzing techniques (which requires the objects to be optimized be programs, not just functions), we show that an alternate form of problem conversion, which is converting program $P$ to $P'$ defined below:

\begin{algorithm}[!h]
    \caption{$P'$: The Converted Program from $P$}
    %% \label{alg:AOA}
    \renewcommand{\algorithmicrequire}{\textbf{Input:}}
    \renewcommand{\algorithmicensure}{\textbf{Output:}}
    \begin{algorithmic}[1]
        \REQUIRE $q: \{0,1\}^M$
        \STATE Symbolically execute program $P$, get $\mathrm{sat(q)}, \mathrm{cost}(q)$
        \IF{$\mathrm{sat}(q)$}
            \STATE AddCost($\mathrm{cost}(q)$)
        \ELSE
            \STATE ExitWithError()
        \ENDIF
        
    \end{algorithmic}
\end{algorithm}

It's clear that to find the worst case (and non-error) input of $P'$ is to find the worst case path string in $P$.
Thus, any common fuzzing techniques could be used.

But why would we do fuzzing on this transformed program instead of fuzzing the original program $P$ itself? Intuitively, two properties of the input space of the new program make it simpler for analysis:
\begin{itemize}
    \item
    For most of the programs, the input space of program $P'$ ($\{0,1\}^M$) is smaller than the input space of program $P$.
    This inherits the benefit from symbolic execution based methods. 
    Although the size of the input space may still be too large, the evolutionary algorithm may utilize this size reduction and explore the input space.
    \item 
    Different programs' inputs may contain variables of different types and shapes (arrays of different dimensions, etc.).
    Designing fuzzing techniques on program $P$ requires the system to be adaptable to different kinds of input space. However, the input of program $P'$ is always binary strings of a fixed length.
    Thus, we can design fuzzing techniques based on this shape of input space.
\end{itemize}

\subsection{Evolutionary Algorithm}
\label{sec:evo-algo}

In this section, we will introduce the fuzzing technique used in this work and focus on designing the evolutionary algorithm.

The target of fuzzing is program $P'$, whose input space is $q \in \{0,1\}^M$, and the output contains $\mathrm{sat}(q), \mathrm{cost}(q)$.
This is a constrained optimization problem, which means we need to find $q$ that first satisfies $\mathrm{sat}(q)$ and then maximizes $\mathrm{cost}(q)$.

It should be noted that although fuzzing is classically a black-box method, which means we know no detail of program $P'$.
In this work, however, we know that program $P'$ is basically the symbolic execution of program $P$, so we cannot make the fuzzing complete black-box.
We may utilize some additional information, for example, the path condition on the execution path $T(q)$.

\subsubsection{Workflow}

The evolutionary algorithm in this work uses the workflow below:

\begin{itemize}
    \item Randomly generate initial inputs and validate their costs;
    \item Repeat the following operations until termination conditions are met (program runtime or iteration count reaches limits):
    \begin{itemize}
        \item Generate offspring using evolutionary algorithm strategies (mutation and crossover operations) based on the current population;
        \item Validate the satisfiability and cost for each offspring;
        \item Select offspring for the next generation population using evolutionary selection strategies.
    \end{itemize}
\end{itemize}
    
 The detailed design of the fuzz testing process corresponds to the configuration of evolutionary algorithm components.
 In the rest of this section, we will provide details on the specific designs of mutation, crossover, and selection operations in our methodology.

\subsubsection{Population and Individual}

In this work, an individual is a path string $q \in \{0,1\}^N$.
The performance of an individual is defined as:

$$
\mathrm{perf}(q) := \begin{cases}
    \mathrm{cost}(q) & \text{if } \mathrm{sat}(q), \\
    -1 & \text{if } \neg\mathrm{sat}(q)
    \end{cases}
$$

It should be noted that this transforms the constrained optimization problem back to a non-constrained optimization problem.

The size of the population is defined by the hyper-parameter $\texttt{psize}$. The initial population is generated by two strategies:

\begin{itemize}
    \item Generate a path string where each bit is assigned $0$ or $1$ at equal probability.
    \item Generate a random input $\alpha \in I_P$, find a related path string $q \in Q(\alpha)$.
\end{itemize}

The second strategy ensures that the generated path string is satisfiable. 
However, it requires more running time and may be biased by the program features (for example, all inputs sampled are related to the same path string).
Considering these facts, we utilize both strategies and let each generate half of the initial population.

\subsubsection{Mutation}

The mutation operation can be considered a randomized function that transforms an individual to another: $\mathrm{mut}(q_1) = q$.
It provides the ability of local search to the evolutionary algorithm.

In this work, we use two different types of mutation operation:

\begin{itemize}
    \item Choose a random bit in $q_{1[1...m]}$, flip that bit to get $q$.
    \item Choose a random bit in $q_{1[1...m]}$, flip that bit, and randomize all bits after it to get $q$.
\end{itemize}

Here, $m$ refers to the last bit that is ever used (so that the program terminates after the symbolic execution requires the $(m+1)$-th bit). Modification of bits after the $m$-th results in an equivalent path string, which is useless.

\subsubsection{Cross-Over}

The crossover operation can be considered a randomized function that takes multiple individuals (usually two) as input and outputs another.
With this operation combining properties of different individuals, the new individuals generated are expected to achieve higher performance.

In this work, we use three different cross over operations $\mathrm{cross}(q_1, q_2) = q$:

\begin{itemize}
    \item Replace a suffix of $q_1$ with a suffix of $q_2$.
    \item Replace a substring of $q_1$ with a substring of $q_2$.
    \item Insert a substring of $q_2$ into a position in $q_1$.
\end{itemize}

All positions are randomly selected between $1 \dots m_1, 1 \dots m_2$, for the same reason in the mutation operation, ensuring that the suffix or substring selected are used in symbolic execution.

\subsubsection{Selection}

In each iteration, a number of mutation and crossover operations generate the next generation of individuals.
A selection strategy selects a certain number of individuals from the last generation and all individuals generated in this iteration to form the next generation, while the rest of them are killed.
To make the method more flexible, reducing the probability of getting stuck in local optima, we use the following strategy:

\begin{itemize}
    \item $R_1 \cdot \texttt{psize}$ individuals are randomly selected from the last generation. Specifically, the individual with the highest performance is always selected.
    \item For all of the generated individuals, the top $R_2 \cdot \texttt{psize}$ are selected.
    \item For the rest of them, randomly select $(1- R_1-R_2)\cdot \texttt{psize}$ individuals. For an individual with performance ranking $i$ and crowdingness factor $j$, its weight in the random selection is $i^{-\beta}j^{-\gamma}$.
\end{itemize}

The concept of using the crowdingness factor is inspired by a well-known multi-objective evolutionary algorithm, NSGA-II~\cite{nsga2}.
In this work, the crowdingness factor is defined as follows:

\begin{definition}[The Crowdingness of Individual]
    For an individual $q$ in the current set of generated individuals $Q$, its crowdingness is:
    $$
    \mathrm{crowd}(q) := \sum_{q'\in Q, q' \ne q} \mathrm{sim}(q, q') ,
    $$
    where $\mathrm{sim}(q, q')$ is defined as the normalized value of the longest common prefix between $q, q'$:
    $$
    \mathrm{sim}(q, q') := (|q| \cdot |q'|)^{-1/2} \max_{0\le i \le \min \{|q|,|q'|\}\}} i \ \mathrm{ s.t.} \forall j\le i, q_j = q_j' .
    $$
\end{definition}

Here $R_1, R_2, (1 - R_1 -R_2) \in [0, 1], \beta,\gamma \in (0, +\infty)$ are hyper-parameters.
Increasing $R_1$ keeps more individuals unchanged between generations, allowing old individuals to survive multiple iterations.
Increasing $R_2$ or $\beta$, or decreasing $\gamma$ makes the convergence slower while lowering the risk of getting stuck in local optima.

%\subsubsection{Summary}
%
%In this subsection, we have introduced the design details of the evolutionary algorithm specifically designed for program $P'$ with binary strings as inputs. Thus we have provided a fuzzing framework to solve the WCA problem on program $P'$, consequently solving the WCA problem on the original program $P$.

\subsection{Optimization}
\label{sec:optimization}

%After designing the method, we conducted preliminary tests on several benchmark programs.
%%
%The results show that \pfuzz{} performs well on some programs but not on others.

The base algorithm of \pfuzz{} may yield a drastic performance difference across different programs.
The main reason for this performance difference lies in the proportion of satisfiable path strings $\mathrm{sat}(q)$ among all path strings.
For certain programs, most or even all paths $q$ satisfy the constraints, causing the transformed optimization problem to degenerate from a constrained optimization problem to a (near) unconstrained one, which can be efficiently solved by evolutionary algorithms.
However, for other programs where a high proportion of unsatisfiable path strings, most individuals generated through mutation or crossover in evolutionary algorithms should be discarded. 
We illustrate this phenomenon using the QuickSort program. The previous QuickSort program (shown in \cref{alg:AOA}) is indeed a program that could be easily analyzed by \pfuzz{}, but what if we make a tiny modification, by initializing array $mid$ with empty array instead of an array containing the pivot, and start the loop from $0$ instead of $1$:

\begin{algorithm}[!h]
    \caption{QuickSort (Simple Impl., Modified) (N)}
    %% \label{alg:AOA}
    \renewcommand{\algorithmicrequire}{\textbf{Input:}}
    \renewcommand{\algorithmicensure}{\textbf{Output:}}
    \begin{algorithmic}[1]
        \REQUIRE $N, A: \texttt{Array}[\texttt{int}[1, N]] \ \mathrm{where} |A| = N$

        \IF{$N \le 1$}
            \RETURN $A$
        \ENDIF

        \STATE AddCost($N$)

        \STATE $pivot \leftarrow A[0]$
        \STATE $left, mid, right \leftarrow [], [], []$

        \FOR{$i \in 0 \dots N-1$}
            \IF{$A[i] < pivot$}
                \STATE $left \leftarrow left + [A[i]]$
            \ELSIF{$A[i] = pivot$}
                \STATE $mid \leftarrow mid + [A[i]]$
            \ELSE
                \STATE $right \leftarrow right + [A[i]]$
            \ENDIF
        \ENDFOR

        \RETURN Quicksort($|left|, left$) $+ mid +$ Quicksort($|right|, right$)
        
    \end{algorithmic}
\end{algorithm}

It's obvious that the behavior of the program is equivalent to the original version. The new program costs a tiny amount of extra time (two comparisons for each recursive function call), but a waste of resources at this scale is often ignored in real-world programming. However, for \pfuzz{}, the modified version is much more difficult to analyze than the original one. What's wrong?

In the above program, during each recursive call to QuickSort, when executing the first iteration of the For loop (where $i=0$), the two conditional checks in the loop become $A[0]<A[0]$ and $A[0]=A[0]$, which respectively evaluate to constant values $\texttt{False}$ and $\texttt{True}$. 
This means that during symbolic execution along path string $q$, the next two bits of $q$ must be $\texttt{"01"}$ when reaching this point; otherwise, $q$ becomes unsatisfiable.
Since the program makes recursive calls, the For loop may execute up to $N-1$ times during QuickSort operations on arrays of length $N$, with each call constraining two bits in path $q$.
Consequently, only $1/2^{2N-2}$ of all possible paths satisfy the constraints.
When $N$ is large, even using mutation and crossover strategies in evolutionary algorithms is likely to generate numerous unsatisfiable paths.
However, the original version of QuickSort does not have the problem. In fact, there is no unsatisfiable path in the original version of QuickSort, and the optimization problem degenerates to an unconstrained one.

Similar situations also occur in path conditions composed solely of concrete values.
For instance, if we implement the $\mathbf{for}$ loop in QuickSort using C-style syntax: $\mathbf{for}$($i=0$;$i<N$;$i++$), then during a single function call, the $i<N$ condition must evaluate to \texttt{True} for the first $N$ checks and False at the $N+1$-th check.
This occurs because variables $i$ and $N$ always hold concrete values in this program (evidently, worst-case analysis here does not require considering different values of $N$).
Thus, the expression $i<N$ always evaluates to concrete Boolean values, which are treated as constants by symbolic executors.

The above observations inform us that to run \pfuzz{} efficiently, the program must be well-organized enough to avoid any unnecessary conditional branch. This is definitely unacceptable for a WCA solver looking forward to working on real-world programs. The optimization below must be applied to make the method complete and applicable.

\subsubsection{Ignore Unsatisfiable Branches}

Consider the case where the program executes through some execution path and is currently at the end of statement $t_i$.
The statement is a branch statement, and its path condition on either side can not be evaluated to a concrete value. 
However, if we take one of the branches, the current path condition may become unsatisfiable. 
If so, we can simply take the other branch without reading the next bit of $q$.
With this idea, we provide the following definition:

\begin{definition}[Path String Mapping Ignoring Unsatisfiable Branches]
     For program $P$, the path string mapping ignoring unsatisfiable branches $T_u(q) = t_1, t_2, \dots, t_l$ is defined as follows:
     
    \begin{itemize}
        \item $t_1 = c_1$.
        \item Iteratively, $\forall i \ge 2:$
        \begin{itemize}
            \item If $|\mathrm{OUT}(t_{i-1})| = 0$, then $t_i = \mathrm{NULL}$ .(No such statement, the program terminated after statement $i-1$.)
            \item If $|\mathrm{OUT}(t_{i-1})| = 1$, then $t_i = \mathrm{OUT}(t_{i-1})_1$
            \item If $|\mathrm{OUT}(t_{i-1})| = 2$, then:
            \begin{itemize}
                \item If adding the path condition of $\mathrm{OUT}(t_{i-1})_1$ to the current set of path conditions would make it unsatisfiable, then $t_i = \mathrm{OUT}(t_{i-1})_2$.
                \item Otherwise, if adding the path condition of $\mathrm{OUT}(t_{i-1})_2$ to the current set of path conditions would make it unsatisfiable, then $t_i = \mathrm{OUT}(t_{i-1})_1$.
                \item Otherwise,
                $$
                t_i = \begin{cases}
                \mathrm{OUT}(t_{i-1})_1 & \text{if } q_{x} = 0, \\
                \mathrm{OUT}(t_{i-1})_2 & \text{if } q_{x} = 1
                \end{cases}
                $$
                Where $x$ is the next unused bit in $q$. (Increment $x$ by 1 in this case.)
            \end{itemize}
            
        \end{itemize}
    \end{itemize}
    
\end{definition}

This optimization, in fact, completely eliminates unsatisfied path strings. We can prove the following theorem:

\begin{theorem}[Complete Satisfiability]\label[theorem]{thm:complete-sat}
     For program $P$, using the path string mapping ignoring unsatisfiable branches $T_u(q)$, for all path string $q$, we have $\mathrm{sat}(q)$.
\end{theorem}
\begin{proof}
    We inductively prove the following theorem: For $0 \le i \le l$ where $l = |T_u(q)|$, the conjunction of all path conditions in the $i$-prefix of $T_u(q)$ is satisfiable. Induction on $i$.
    \begin{itemize}
        \item $i = 0$. The conjunction of the empty set is $\texttt{True}$, which is satisfiable.
        \item $i > 0$. By induction, the path conditions in the $i-1$-prefix are satisfiable. Consider $\mathrm{OUT}(t_{i-1})$:
        
        \begin{itemize}
            \item $|\mathrm{OUT}(t_{i-1})| = 1$. As it's a non-conditional jump, the path condition at $t_i$ is $\texttt{True}$; thus, adding it to the set of the path condition of the first $i-1$ statements (already satisfiable) is also satisfiable.
            
            \item $|\mathrm{OUT}(t_{i-1})| = 2$. By definition of $T_u(q)$, if any branch of $t_i$ makes the set of path condition unsatisfiable, it will select the other branch.
            
        \end{itemize}
    \end{itemize}
\end{proof}

With this optimization, the problem that we need to solve degenerates from a constrained optimization problem to an unconstrained one, making every new path explored more useful than the previous methods.
However, it should be noted that this optimization is not free and requires much extra time budget from \pfuzz{} itself. It requires a call to the SMT solver at every branch where the path condition is a symbolic value, which is expensive in large programs.

\subsubsection{Skip Satisfiability Check}

Based on the above optimization, we then focus on reducing calls to the SMT solver. We notice that in many programs, some branches satisfy: for all possible inputs and control flows that reach this branch, the path conditions for both sides are always satisfiable.

The above observation provides an intuition to an optimization. If we can identify these "always satisfiable" branches, then all calls to the SMT solver could be skipped at these branches. With this optimization, the number of calls to the SMT solver could be reduced by a portion (depending on the program itself), reducing time cost evaluating each path, improving efficiency of the method.

In extreme cases, a program could consists of purely two types of branches easy for analysis: a branch either has a concrete path condition (e.g. $\texttt{True}$ or $\texttt{False}$, thus satisfiability of both sides could be checked without SMT solver), or has a symbolic path condition that both sides is satisfiable. If so, all calls to the SMT solver can be eliminated, even including checking the satisfiability of the entire path $q$ at the end of symbolic execution.

Then, how do we identify such a type of branch? This could be done with one of the following methods:

\begin{itemize}
    \item Manual Instrumentation. For simplicity, we can put special instrumentation in the programs to infer \pfuzz{} that satisfiability checks could be skipped at some branches.
    \item Testing. For real-world programs, it's inefficient to manually analyze the satisfiability at each branch.
\end{itemize}

\begin{comment}
    
\subsubsection{Completely Satisfiable Programs}

We have noticed that many textbook algorithms in our test benchmark do not need the path string mapping ignoring unsatisfiable branches.
%
After ignoring concrete path conditions, these programs already satisfy the condition below:

\begin{definition}[Completely Satisfiable Program]
    A Program $P$, under some path string mapping function $T(q)$ is completely satisfiable, if $\forall q\in \{0,1\}^M, \mathrm{sat}(q)$.
\end{definition}

\cref{thm:complete-sat} shows us that if $T(q) = T_u(q)$, then all programs are completely satisfiable. If not stated otherwise, in the following part, we consider the case $T(q) = T_c(q)$.

If we know program $P$ is completely satisfiable, then the problem again degenerates to unconstrained optimization, making path strings generated more useful.
%
Besides that, all satisfiability check of path strings, i.e., check for $\mathrm{sat}(q)$, can be completely removed, which means we do not need to call the SMT solver at all until the end of evolutionary algorithm, where we still need the SMT solver to find the concrete input of our worst case path string.
%
Such properties make our method run more efficiently on these programs.
\end{comment}

\section{Evaluation Setup}
\label{sec:evaluation-setup}

In this section, we will introduce our evaluation setup, both theoretically (the research questions) and practically (how we implemented the method).

\subsection{Research Questions}

Before evaluation, we propose the research questions below.

\paragraph{RQ1} Is \pfuzz{} effective and efficient in solving WCA problems on general programs?

\paragraph{RQ2} What is the advantage of \pfuzz{}? On what type of programs could \pfuzz{} exhibit exceptional efficiency?

\paragraph{RQ3} How effective is \pfuzz{} compared to classic methods based on fuzzing or symbolic execution?

\subsection{Implementation}

\begin{table*}
    \caption{\textbf{All Benchmark Programs}}\label{tab:benchmarks}
    \centering
    \begin{tabular}{cccc}
    \toprule
    ID & Test Program & Description & Input Scale \\
    \midrule 
    1-1 & InsertionSort & Insertion sort algorithm & $N=128$ \\
    1-2 & QuickSort & Simple implementation of QuickSort & $N=128$ \\
    1-3 & HeapInsertion & Inserting elements into a min-heap & $N=128$ \\
    1-4 & Dijkstra & Single-source shortest path algorithm & $N=8$ \\
    1-5 & BSTInsertion & Inserting elements into a binary search tree &$N=128$ \\
    1-6 & BellmanFord & Single-source shortest path algorithm & $N=8$ \\
    1-7 & BellmanFordQueue & Queue-optimized Bellman-Ford algorithm & $N=8$ \\
    1-8 & HashTable & Inserting elements into a hash table & $N=8,P=13$ \\
    \midrule 
    2-1 & InsertionSort' & Badger's version for InsertionSort & $N=64$ \\
    2-2 & QuickSortJDK & QuickSort algorithm from JDK 1.5 & $N=64$ \\
    2-3 & HashTable' & Realistic hash table implementation & $N=64$ \\
    2-4 & ImageProcessor & Image processing functions & $SX,SY=2$\\
    2-5 & SmartContract & Smart contract for cryptocurrency &  $N=50$\\
    \midrule 
    3-1 & IsPalindrome & Checks if a string is palindrome & $N=100$ \\
    3-2 & IsPalindrome' & Optimized version of IsPalindrome & $N=100$ \\
    3-3 & MemoryFill & Copy non-zero elements between lists & $N=100$ \\
    3-4 & Alternate0 & Cost maximized when inputs alternate 0/non-0 & $N=100$ \\
    3-5 & DFS & Depth-first search on n-node graph & $N=10$ \\
    3-6 & BFS & Breadth-first search on n-node graph & $N=10$ \\
    \bottomrule
    \end{tabular}
\end{table*}

\paragraph{\pfuzz{}}
To conduct experiments and test our method, we have implemented a simple (simulated) symbolic execution engine in Python.
We write programs as Python generator objects to communicate explicitly with the symbolic executor.
The program iteratively yields the path conditions when a branch is encountered and receives the selected concrete path from the symbolic executor.
The program also yields cost information to let the symbolic executor accumulate its resource consumption.
% (via an AddCost function)
 
Besides the symbolic executor, we have also implemented the evolutionary algorithm designed in the above section in Python.
This algorithm calls the symbolic executor to compute the cost of path strings, i.e., the performance of individuals.

When we combine them, we get a full implementation of \pfuzz{}.
Given a program, along with its resource definition and input space, \pfuzz{} will take some time and return a set of inputs believed to be the worst case (or an under-approximation) along with its resource consumption.

\paragraph{Methods for Comparison} We have also implemented some classic methods in the same environment:

\begin{itemize}
    \item Fuzzing: an evolutionary-algorithm based fuzzing method similar to the implementation of KelinciWCA, a WCA-specialized version of Kelinci, based on AFL~\citep{MISC:afl}, described in Badger~\cite{ISSTA:NKP18}.
    \item SymExe: a search-based symbolic execution method. The search strategy is based on the strategy of JPF-Symbc\cite{spf}, also described in Badger.
\end{itemize}

\subsection{Benchmark Programs}

All benchmark programs used in the experiments are listed in Table \ref{tab:benchmarks}.
The benchmark programs of our experiments consist of the follows:

\begin{itemize}
    \item 9 programs (1-1 $\sim$ 1-9) selected by ourselves. All of them are classic textbook algorithms.
    \item 5 programs (2-1 $\sim$ 2-5) ported from the experiments of Badger.
    \item 6 programs (3-1 $\sim$ 3-6) ported from the experiments of ESE.
\end{itemize}

Programs suffixed with a prime have the same algorithm with the respective non-suffixed programs, but are alternate implementations.

The set of benchmark programs has the following properties:
\begin{itemize}
    \item SmartContract is a program with symbolic cost. The other 19 programs has concrete cost.
    \item QuickSort', HashMapInsertion', and ImageProcessor are programs with a relatively large scale ($\sim 100$ lines of code). The other 17 programs are relatively small ($<50$ lines of code).
\end{itemize}

\section{Evaluation Result}
\label{sec:evaluation-result}

The hardware and software environment of the experiment is as follows:

\begin{itemize}
    \item CPU: 12th Gen Intel(R) Core(TM) i9-12900H   2.50 GHz
    \item Memory: 32 GB
    \item OS: Windows 10
    \item Python 3.10.11, z3-solver 4.12.2.0
\end{itemize}

For each benchmark program, we run all three methods at least four times and record their average performance. Each time, we run the method for at most 120 minutes.

The evaluation results on all 19 benchmark programs are shown in Table \ref{tab:eval_results}. The numbers represent the cost of the worst-case input found. Bold numbers mark the methods that perform the best in each benchmark program. The minimum time to find the worst case is compared and listed as subscripts if multiple methods find the same maximum cost (usually the case that all of them have found the theoretical maximum cost).

The results show that among all 19 benchmark programs, \pfuzz{} has outperformed fuzzing in 13 programs, symbolic execution in 17 programs, and both of the two classical methods in 12 programs.

Figure \ref{fig:exp} shows experiment results on some of the benchmark programs, which are representative of \pfuzz{}'s advantages. Figures for other benchmark programs can be found in supplementary materials.

\begin{table*}
    \caption{\textbf{Benchmark programs. The statistics show the resource costs of the found worst-case inputs by the methods.}}
    \label{tab:eval_results}
    \centering
    \begin{tabular}{cccc|cc}
    \hline
    ID & Test Program & \pfuzz{} (cost) & Path String (length) & Fuzzing (cost) & SymExe (cost) \\
    \hline 
    1-1 & InsertionSort & $7356$ & 16266 & $\textbf{7945}$ & $1984$  \\
    1-2 & QuickSort & $\textbf{8255}$ & 165110 & $4842$ & $954$ \\
    1-3 & HeapInsertion & $\textbf{649}$ & 906 & $585$ & $276$  \\
    1-4 & Dijkstra & $\textbf{28.0}$ & 1034 & $27.8$ & $25.3$ \\
    1-5 & BSTInsertion & $1355$ & 16394 & $\textbf{5976}$ & $1529$ \\
    1-6 & BellmanFord & $\textbf{485}$ & 1034 & $482$ & $384$  \\
    1-7 & BellmanFordQueue &$\textbf{27.6}$ & 1034 & $27.0$ & $17.8$ \\
    1-8 & HashTable & $28.0_{80}$ & 142 & $\textbf{28.0}_{0.16}$ & $24.5$ \\
    \hline 
    2-1 & InsertionSort' & $\textbf{4222}$ & 4042 & $4189$ & $1449$  \\
    2-2 & QuickSortJDK & $1018$ & 32778 & $\textbf{1262}$ & $942$ \\
    2-3 & HashTable' & $4348$ & 586 & $\textbf{4991}$ & $2988$  \\
    2-4 & ImageProcessor & $\textbf{36211}$ & 200 & $33238$ & $2109$ \\
    2-5 & SmartContract & $21474938_3$ & 100 & $\textbf{21474938}_{0.04}$ & $21474938_4$ \\
    \hline 
    3-1 & IsPalindrome & $100_{0.08}$ & 210 & $100_{25}$ & $\textbf{100}_{0.04}$  \\
    3-2 & IsPalindrome' &$\textbf{50}_{0.019}$ & 210 & $50_{24}$ & $50_{0.022}$  \\
    3-3 & MemoryFill & $\textbf{100}_{0.04}$ & 210 & $100_{0.35}$ & $21$  \\
    3-4 & Alternate0 & $\textbf{600}_{0.32}$ & 210 & $600_{13}$ & $409$  \\
    3-5 & DFS & $\textbf{100}$ & 130 & $99$ & $33$ \\
    3-6 & BFS & $\textbf{100}$ & 130 & $97$ & $33$ \\
    \hline
    \end{tabular}
\end{table*}

\begin{figure*}[h!]
    \centering
    \hfil
    \begin{subfigure}[b]{0.36\linewidth}
        \includegraphics[width=\linewidth]{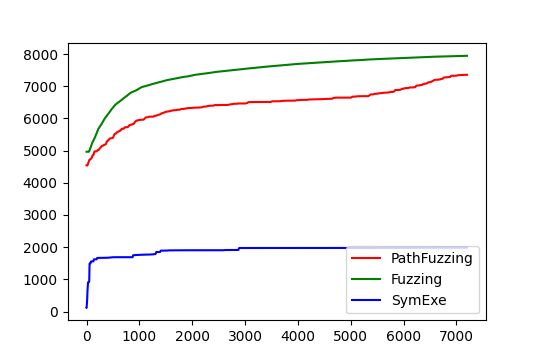}
        \caption{1-1-InsertionSort}
    \end{subfigure}
    \hfil
    \begin{subfigure}[b]{0.36\linewidth}
        \includegraphics[width=\linewidth]{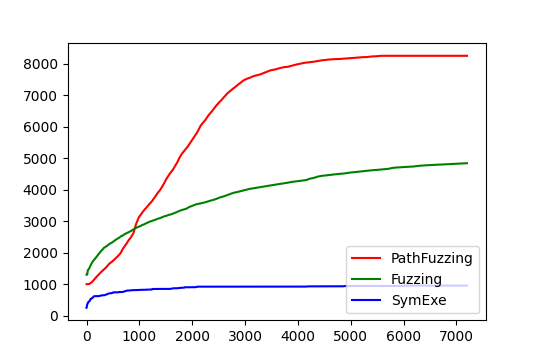}
        \caption{1-2-QuickSort}
    \end{subfigure}
    \hfil
    \\
    \hfil
    \begin{subfigure}[b]{0.36\linewidth}
        \includegraphics[width=\linewidth]{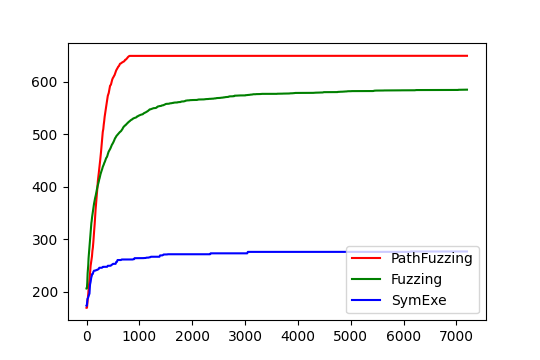}
        \caption{1-3-HeapInsertion}
    \end{subfigure}
    \hfil
    \begin{subfigure}[b]{0.36\linewidth}
        \includegraphics[width=\linewidth]{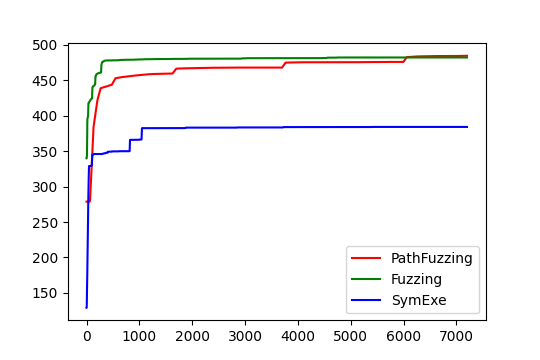}
        \caption{1-6-BellmanFord}
    \end{subfigure}
    \hfil
    \\
    \hfil
    \begin{subfigure}[b]{0.36\linewidth}
        \includegraphics[width=\linewidth]{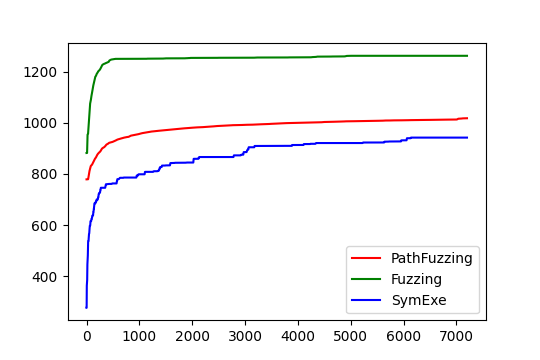}
        \caption{2-2-QuickSortJDK}
    \end{subfigure}
    \hfil
    \begin{subfigure}[b]{0.36\linewidth}
        \includegraphics[width=\linewidth]{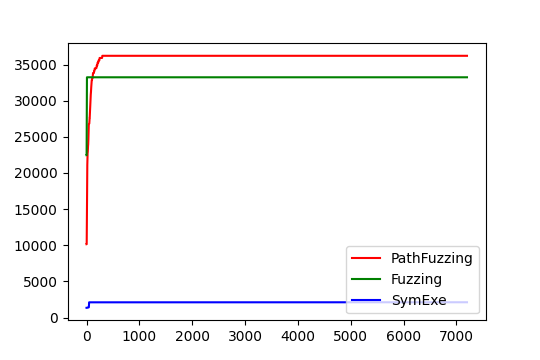}
        \caption{2-4-ImageProcessor}
    \end{subfigure}
    \hfil
    \caption{Results on a subset of benchmark programs.}
    \label{fig:exp}
\end{figure*}

\section{Discussion}
\label{sec:discussion}

\paragraph{RQ1} The experimental results demonstrate that \pfuzz{} achieves certain efficiency and effectiveness in solving WCA problems. 

On all 19 programs, \pfuzz{} is effective: inputs that significantly slow down programs have been found, even for the most difficult WCA problem. \pfuzz{} has overcome one of the weaknesses of search-based symbolic execution, which sometimes has nearly no effects on complex programs and large-scale inputs.

Specifically, on some of the programs, \pfuzz{} quickly found near-optimal worst case input in 10 minutes, suggesting potential usability in real-world software analysis tasks. 

\paragraph{RQ2} Experiment results show that \pfuzz{} has the most outstanding performance on small-scale programs with large-scale inputs. Evolutionary algorithms are believed to rely on mass individuals and thus require fast evaluation for each individual. Evaluating complex path conditions requires intolerable time calling SMT solver. 

The proportion of branches whose satisfiability check could be skipped is a key factor to the efficiency of \pfuzz{}, as it also greatly impacts the number of calls to SMT solver required. 

For programs with small-scale inputs, \pfuzz{} would also work well, but common symbolic execution based methods may have time enumerating all paths and perform even better.

\paragraph{RQ3} On our benchmark programs, \pfuzz{} has a slight advantage compared to common fuzzing techniques and has a significant advantage compared to search-based symbolic execution techniques.

Through experiments, we have confirmed that \pfuzz{} has significantly reduced the input space on some programs. Take IsPalindrome' as an example; this program has an input space of size $256^{100}$ but only has $50$ distinct paths. This makes \pfuzz{} perform better than fuzzing in some of the programs.

Compared to symbolic execution, we believe that \pfuzz{} inherits the balance between global and local search, which is considered an advantage of the evolutionary algorithm. The input scale of our benchmark programs is a little bit too large for search-based methods.

\begin{comment}
However, when the program complexity grows further, \pfuzz{} exposes one of its weaknesses. Evolutionary algorithms is believed to rely on mass individuals, which requires fast evaluation for each individual. 
On complex programs like QuickSortJDK and both versions of HashTable, the time gap between symbolic execution and concrete execution becomes huge. Thus, within certain amount of time, \pfuzz{} can only evaluate like $1/10000$ as many individuals as classic fuzzing, result in poor performance.

Another drawback of \pfuzz{} is that it sometimes get stuck in local optima. If we compare their best performance through multiple experiments instead of average, then for program InsertionSort, \pfuzz{} will outperform fuzzing. For other programs, there is a great increase.
\end{comment}

\paragraph{Threats to Validity}

Our study contains several potential threats to validity.
First of all, although we have selected a wide range of benchmark programs, this may be insufficient to cover different types of programs. An obvious problem is that we lack true real-world scale programs (like thousands of lines).
In implementing our method, we utilized some extra manually generated information. The path length $M$ is taken from user input, and whether a sat check could be skipped in a branch is decided by manual instrumentation in the programs. Although we have demonstrated that these can be replaced by automatic processes, this difference may slightly impact performance.
As for the experiment, the performance of each method may be sensitive to the parameters. Most of the parameters remain constant across experiments; however, for \pfuzz{}, we sometimes use different values of population size according to program specialties.

\section{Related Work}
\label{sec:related}

%%\paragraph{Badger}
%%The work by \citet{ISSTA:NKP18}, has provided a compound WCA method based on fuzzing and symbolic execution. The core idea of this work is to run a symbolic execution engine and a fuzzing engine in parallel on the same program, while continuously exchanging inputs between the two processes. The experiments have been done on several programs in Java, using KelinciWCA (a fuzzing tool for Java based on AFL) and Symbolic Pathfinder (a symbolic execution tool for Java). In most cases, Badger outperformed individual use of either components.

%%\paragraph{Evolutionary Symbolic Execution (ESE)}
%%The work by \citet{ISSRE:ADS18} is the most closely related to the method discussed in this paper. This approach treats all satisfiable paths of a program as a search space and uses evolutionary algorithms to search for worst-case paths within it. Unlike this paper, in this work a path is represented as a sequence of path conditions. In the crossover operation, offspring are generated by combining random halves of the parents. Symbolic execution is then applied to select a random path that satisfies the new path condition sequence. If there is none, the offspring is discarded.

\paragraph{WCA by Combination of Fuzzing and Symbolic Execution}
Many previous works have explored the combination of these two methods.
Badger \citep{ISSTA:NKP18} is a foundational work, which runs a symbolic execution engine and a fuzzing engine in parallel while continuously exchanging inputs. The experiments have been done on several programs in Java, using KelinciWCA (based on AFL~\cite{MISC:afl}) and Symbolic Pathfinder~\cite{spf}.
ESE \citep{ISSRE:ADS18} runs an evolutionary algorithm on program paths defined by sequences of path conditions. Symbolic execution is then applied to select a concrete input that satisfies the newly generated path condition sequence.
CSEFuzz \citep{xie2020csefuzz} uses symbolic execution to generate a set of inputs with high code coverage as initial inputs for fuzzing.
HyDiff \citep{noller2020hydiff} identifies differences in program behavior under different inputs.
Another work by \citet{gerasimov2018combining} uses symbolic execution to improve code coverage in fuzzing.

\paragraph{WCA by Fuzzing}
The concept of fuzzing was first introduced by \citet{miller1990empirical} in 1990.
Early applications such as AFL \citep{MISC:afl}, LibFuzzer \citep{MISC:libfuzzer}, and Radamsa \citep{MISC:radamsa} primarily focused on identifying software defects and vulnerabilities.
In recent years, many methods have been proposed for solving worst-case analysis using fuzzing.
General approaches include WISE \citep{ICSE:BJS09} by \citet{ICSE:BJS09} and SlowFuzz \citep{CCS:PZK17} by \citet{CCS:PZK17}.
Specialized techniques such as \citet{NDSS:BMA20} perform mutation and crossover operations on Java objects. 
\citet{le2021saffron} requires users to provide the syntax of program inputs and generates only syntax-compliant inputs.
Singularity \citep{wei2018singularity} optimizes input patterns instead of inputs themselves to find worst-case inputs for different data scales.

\paragraph{WCA by Symbolic Execution}
Symbolic execution is a well-established technique dating back to \citet{SAM:Floyd67}'s 1967 paper, originally intended for proving algorithm correctness.
Classic works including KLEE \citep{OSDI:CDE08}, DART \citep{PLDI:GKS05}, and EXE \citep{cadar2008exe}.
In recent years, \citet{liu2022acquirer} first uses static program analysis to identify code fragments potentially causing algorithmic complexity issues, then applied symbolic execution to these fragments; \citet{kebbal2006combining} used control flow graphs to analyze block execution counts analytically, avoiding loop unrolling; \citet{bernat2000approach} used a computational algebra system to avoid subroutine calls; and \citet{liu1998automatic} improved the efficiency of symbolic execution through partial evaluation and incremental computation.
Unlike fuzzing, some symbolic execution works focus solely on theoretical analysis without experimental validation.

\paragraph{Other WCA Methods}
Besides fuzzing and symbolic execution, many other methods can be applied to solve worst-case analysis problems, such as machine learning \citep{kumar2020detection}, surrogate models \citep{mukkamala2003detecting}, and type-system-based approaches \citep{POPL:WH19}.

\section{Conclusion}
\label{sec:conclusion}

In this work, we present \pfuzz{}, a worst-case analysis method combining symbolic execution and fuzzing techniques.
We present the complete framework of the method along with key optimization strategies and conducted experiments based on Python implementations.
Our experimental results demonstrate that the method achieves satisfactory optimization effects for most programs and shows certain advantages over classic fuzzing and symbolic execution techniques.
%

%%0

%% No Acknowledgements

%% \begin{acks}
%% ACK?
%% \end{acks}

\bibliographystyle{ACM-Reference-Format}
\bibliography{db}


\end{document}